\documentclass[aps,pra,twocolumn,10pt,superscriptaddress,floatfix,longbibliography]{revtex4-1}

\usepackage[T1]{fontenc} 
\usepackage[english]{babel}

\usepackage[utf8]{inputenc}
\usepackage{bbold}

\usepackage[table,usenames,dvipsnames]{xcolor}
\usepackage[colorlinks=true, hypertexnames=false]{hyperref}
\usepackage{amsmath}
\PassOptionsToPackage{linktocpage}{hyperref} 
\usepackage{amsthm} 
\usepackage{amssymb}
\usepackage{amsmath,mathtools}
\usepackage{dsfont}
\usepackage[vcentermath]{youngtab}

\usepackage{tikz}
\usetikzlibrary{matrix}
\usetikzlibrary{arrows}

\usepackage{paralist}

\usepackage{graphicx}

\graphicspath{{./pics/}{./}}

\newtheorem{theorem}{Theorem}

\newtheorem{lemma}[theorem]{Lemma}

\newtheorem{definition}[theorem]{Definition}
\newcommand{\myleft}{\mathopen{}\mathclose\bgroup\left}
\newcommand{\myright}{\aftergroup\egroup\right}
\DeclareMathOperator\tr{Tr}

\newcommand{\ket}[1]{\left.\left|{#1}\right.\right\rangle}

\newcommand{\bra}[1]{\left.\left\langle{#1}\right.\right|}

\newcommand{\ii}{\mathrm{i}}
\newcommand{\Spec}{\mathrm{Spec}}

\definecolor{jonas}{rgb}{.1,.5,.5}

\makeatletter 
 \hypersetup{pdftitle = {TBD},
	     pdfauthor = {TBD},
	     pdfsubject = {quantum physics},
 	     pdfkeywords = {Wigner
			    }
	    }
\makeatother

\newcommand{\fu}{Dahlem Center for Complex Quantum Systems, Freie Universit\"{a}t Berlin, Germany}

\begin{document}
	
	\title{Equivalence of contextuality and Wigner function negativity\\ in continuous-variable quantum optics}

	\author{Jonas Haferkamp}
	\affiliation{\fu}
	\author{Juani Bermejo-Vega}
    \affiliation{Departamento  de Electromagnetismo y Física de la Materia, Universidad de Granada, 18010 Granada, Spain.}
	\affiliation{Institute Carlos I for Theoretical and Computational Physics, Universidad de Granada, 18010 Granada, Spain.}

	\begin{abstract}
		One of the central foundational questions of physics is to identify what makes a system quantum as opposed to classical.
		One seminal notion of classicality of a quantum system is the existence of a non-contextual hidden variable model as introduced in the early work by Bell, Kochen and Specker.
       In quantum optics, the non-negativity of the Wigner function is a ubiquitous notion of classicality.
		In this work we establish an equivalence between these two concepts.
		In particular, we show that any non-contextual hidden variable model for Gaussian quantum optics has an alternative non-negative Wigner function description. Conversely,  it was known that the Wigner representation provides a non-negative non-contextual  description of Gaussian quantum optics. It follows that contextuality and Wigner negativity are equivalent notions of non-classicality and equivalent resources for this quantum subtheory. In particular, both contextuality and Wigner negativity are necessary for a computational speed-up of quantum Gaussian optics.
		At the technical level, our  result holds true for any subfamily of Gaussian measurements that include  homodyne measurements, i.e., measurements of standard quadrature observables.
	
	\end{abstract}

	\maketitle

	Identifying notions of non-classicality that distinguish genuine quantum systems from classical ones is a central aim in the foundations of quantum physics. 
	Moreover, it is a vital question in quantum computing to delineate which quantum features give rise to computational advantages.~\cite{cleve1998quantum,vidal2003efficient,markov2008simulating,gross2009most,veitch_simulation_2012,galvao2005discrete,mari_wigner_2012,howard_contextuality_2014,raussendorf2013contextuality,stahlke2014quantum,arute2019quantum}. Two central notions of non-classicality are the negativity of the Wigner function~\cite{wigner1997quantum} and quantum contextuality~\cite{kochen1975problem,spekkens_equivalent_2008}. 
	
    The Wigner function is a quasi-probability representation, which is ubiquitous in the field of quantum optics~\cite{case2008wigner,braunstein2005quantum}. If the Wigner function of the sates and operations in a given experiment is non-negative, it provides an efficient probabilistic classical simulation over the phase space for~\cite{veitch_simulation_2012,mari_wigner_2012,veitch2013efficient,delfosse_contextuality_2015,raussendorf_contextuality_2017,bermejo-vega_contextuality_2017,raussendorf_phase-space-simulation_2020}. Thus, Wigner negativity of quantum states and processes is a necessary resource for  quantum computational computing.
	
	A seminal notion of quantumness is the  contextuality of a quantum subtheory, defined as the non-existence of a ``classical" non-contextual hidden variable description description~\cite{kochen1975problem,bell1966problem,mermin1990simple,peres1991two,mermin1993hidden,abramsky2011sheaf}.	The latter are models reproducing quantum mechanics from predetermined assignments of measurement outcomes to observables that do not depend on the specific  measurement context. Similarly to the negativity of the Wigner function, contextuality has been identified as a necessary resource for quantum computational power in several models of quantum computation~\cite{anders_contextuality_2009,raussendorf_contextuality_2013,howard_contextuality_2014,delfosse_contextuality_2015,raussendorf_contextuality_2017,bermejo-vega_contextuality_2017}.
	
    Previous works have asked what is the connection between these two central notions of classicality.	
    Ref.~\cite{spekkens_equivalent_2008} proved that generalized contexuality in the sense of Spekkens is equivalent to the non-existence of a non-negative quasi-probability representation. However, in order to check whether a system is contextual, one would have to consider all possible  quasi-probability distributions. Later, for the concrete model of quantum computation with Clifford operations and magic states on odd-dimensional qudits, Ref.~\cite{howard_contextuality_2014} proved that traditional Bell-Kochen-Specker contextuality is equivalent to the negativity of the discrete Wigner function \cite{gross_hudson_2006,gibbons2004discrete}, using the graph theoretic approach to contextuality~\cite{cabello2014graph}.
	In particular, for single prime-dimensional qudit states, they prove  equivalence between contextuality with respect to Pauli measurements and negativity of the discrete Wigner. 
	The result was simplified and generalized to general multiqudit systems of odd local dimensions in~\cite{delfosse_wigner_2014}.
	
	In this work we show such an equivalence for continuous-variable quantum optical systems. We show that non-contextuality of quantum states with respect to homodyne measurements implies the non-negativity of their Wigner function representation (Theorem~\ref{theorem:main}). This implies the existence of  a non-contextual hidden variable model that for the subtheory of Gaussian quantum optical states, measurements and transformations, in the generalized sense of Spekkens~\cite{spekkens2005contextuality} (Section \ref{sect:FullGaussianOpticsAndSpekkens}). Furthermore, the non-contextual representation based on the non-negative Wigner function is unique. This is proven for the canonical Wigner representation studied in quantum optics~\cite{wigner1997quantum}.	Remarkably, the phase space methods for the discrete-variable case where developed much later than those for continuous variables, but the equivalence was first proven for the discrete case. During the completion of this manuscript, we became aware of an independent proof of our main result~\cite{booth2021contextuality}. Ref.~\cite{booth2021contextuality} provides an analogous proof to  Theorem~\ref{theorem:main} phrased in the continuous-variable extension~\cite{barbosa2019continuous} of the sheaf-theoretic framework for contextuality of Abramsky and Branderburger~\cite{abramsky2011sheaf}.
	
	Our result has several implications. For models of optical quantum computing with continuous variables, negativity of the Wigner function was identified earlier as a necessary resource for a superpolynomial quantum speed-up~\cite{mari_wigner_2012,veitch_simulation_2012}. 	As a consequence, contextuality is necessary to rule out the existence of an efficient classical simulation and for a quantum speed-up in continuous-variable quantum optics. As an example, multiple non-Gaussian  quantum resources that have been studied in the literature are contextual. This includes the Gottesman-Kitaev-Preskill comb state, which is used in fault-tolerant schemes of  continuous-variable quantum computation~\cite{gottesman2001encoding,terhal2020towards} and is a universal resource for Gaussian quantum operations~\cite{baragiola2019all}. Other examples include cubic phase gates \cite{gottesman2001encoding}, which are also universal resources. Additionally, Wigner-negative quantum resources used in quantum advantage setups, such as boson sampling schemes \cite{aaronson_computational_2013} and  variants thereof \cite{lund_boson_2014,chabaud_continuous-variable_2017,kruse_detailed_2019,hamilton_gaussian_2017,quesada_gaussian_2018} are contextual as well: this includes input states such as	photon number states, photon-added or photon-subtracted squeezed states, and measurements such as photon counters, single-photon detectors and threshold detectors.
	
	Combined with our result, Hudson's theorem~\cite{hudson_wigner_1974} identifies the pure states that are contextual with respect to Gaussian operations to be the set of Gaussian states. Hence, non-contextuality is equivalent to Gaussianity for pure states. For mixed states, it was shown that non-negative states can however be non-Gaussian~\cite{veitch_simulation_2012}, and this extends to contextual mixed states via our equivalence. 
    
    Our result also singles out the Wigner negativity of a state or detector as a natural quantifier for contextuality in continuous-variable quantum optics.  The Wigner logarithmic negativity is a computable monotone, which gives rise to a resource theory for quantifying this type of contextuality \cite{albarelli_resource_2018}. Because of Hudsons's theorem,  theories of non-Gaussianity \cite{albarelli_resource_2018,lami_gaussian_2018} could also be employed to quantify contextuality in pure states scenarios.
    Furthermore, these allows one to conceive experimental protocols for measuring, certifying or  witnessing  contextuality using protocols that are available for Wigner negativity ~\cite{nogues2000measurement,kurtsiefer1997measurement,chabaud2021witnessing,banaszek1999direct,hlouvsek2021direct,bertet2002direct}.

	\section{Phase space formulation of continuous variable systems}\label{section:preliminaries}
	We are working with a system of $m$ modes of continuous variables described by a Hilbert space $L^2(\mathbb{R}^m)$.
	With $q_1,...,q_m, p_1,...,p_m$ denoting the usual quadrature operators, we can introduce the displacement operators:
	
	\begin{equation}
	D(\zeta):=e^{\ii(\zeta^{T}\omega\hat{R})}=e^{\ii\sum_ip_i\hat{q}_i-x_i\hat{p}_i}
	\end{equation}
	with $\zeta=(x_1,...,x_m,p_1,...,p_m)\in\mathbb{R}^{2m}$, $\hat{R}:=(\hat{q}_1,...,\hat{q}_m,\hat{p}_1,...,\hat{p}_m)$ and $\omega$ being the matrix representation of the standard symplectic form on $\mathbb{R}^{2m}$:
	\begin{equation}
	\omega:=\begin{pmatrix}
	0&-\mathbb{1}_m\\
	\mathbb{1}_m & 0
	\end{pmatrix}.
	\end{equation}
	Furthermore, we introduce the following notation for the symplectic form on $\mathbb{R}^{2m}$:
	\begin{equation}
	[u,v]\coloneqq u^{T}\omega v.
	\end{equation}
	
	Given a state in form of a density matrix $\rho\in B(\mathcal{H})$, the displacement operator can be used to define its characteristic function $\chi_{\rho}:\mathbb{R}^{2m}\to \mathbb{R}$ via:
	\begin{equation}
	\chi_{\rho}(\zeta):=\tr[\rho D(\zeta)].
	\end{equation}
	
	The \emph{Wigner function} is a function on the phase space that reproduces the predictions of quantum theory. 
	In general, these are functions from the phase space into the real numbers.
	The Fourier transform is naturally defined on $L^1$ in its usual integral form.
	While it can be extended from $L^1\cap L^2$ to a full unitary operator on $L^2$, here we want to work with the integral form of the Fourier transform and therefore restrict to states in this intersection. 
	By further assuming that the Fourier transform of the state is in $L^1$, we can ensure that the Wigner function is always a well defined integrable function:
	\begin{definition}[Wigner function]\label{definition:wignerfunction}
Given a density operator $\rho=\sum_{\alpha}\ket{\psi_{\alpha}}\bra{\psi_{\alpha}}$,such that $\psi_{\alpha} \in L^2(\mathbb{R}^m)\cap  L^1(\mathbb{R}^m)$ and $\mathcal{F}(\psi_{\alpha})\in L^1(\mathbb{R}^m)$, where $\mathcal{F}$ denotes the Fourier transform. 
 We define the \emph{Wigner function} associated to $A$ as the symplectic Fourier transformation of the characteristic function:
		\begin{equation}
		W_{\rho}(\zeta):=\frac{1}{2\pi}\int\chi_{\rho}(v)e^{\ii [v,\zeta]}\mathrm{d} v.
		\end{equation}
	\end{definition}
	For all $\rho$, it is $W_{\rho}\in L^2(\mathbb{R}^{2m})$ (see e.g.~\cite[Thm.~9]{degosson_harmonic_2017}).
The Wigner function is normalized in the sense that 
	\begin{equation}
	\int W_{\rho}(v)\mathrm{d}v=1.
	\end{equation}
	Nevertheless, the Wigner function is a quasi-probability measure because it admits negativ values in general. 
    This negativity is often suggested as a candidate for non-classicality. 
	Another is the abscence of non-contextual value assignments as discussed in the next section.
	
 For the following we need to recall the notion of a spectral measure. 
The quadrature operators do not have a discrete decomposition into projectors onto eigenspaces.
However, by the spectral theorem, one can find for each self-adjoint operator $A:\mathcal{H}\to \mathcal{H}$ on a Hilbert space a projector-valued measure $\Pi_{A}$.
That is a map $\Pi_{A}:B[\mathrm{Spec}(A)]\to \mathrm{P}(\mathcal{H})$, where $B[\mathrm{Spec}(A)]$ is the Borel algebra and $P(\mathcal{H})$ the set of orthogonal projectors acting on $\mathcal{H}$.
$\Pi_{A}$ satisfies direct analogues of the usual axioms for measures.
This allows to define a direct analogue of a projector valued integral over these measures.
The spectral measure $\Pi_A$ satisfies that the map $\mu_{\rho}$ defined by
\begin{equation}
    \mu_{\rho}(X):=\mathrm{Tr}[\rho \Pi_A(X)]
\end{equation}
for a set $X$ is a measure such that the following holds for all integrable functions $f:\mathrm{Spec}(A)\to \mathbb{R}$:
\begin{equation}
    A=\int \lambda \mathrm{d}\Pi_A(\lambda).
\end{equation}
Moreover, it allows to define the following functional calculus:
\begin{equation}
    f(A)=\int f(\lambda) \mathrm{d}\Pi_A(\lambda).
\end{equation}
For expectation values over a state $\rho$ it holds that
\begin{equation}
\mathrm{Tr}\left[\rho A\right]=\int f(\lambda) \mathrm{d}\mu_{\rho}(\lambda),
\end{equation}
which we also denote by $\int f(\lambda) \mathrm{Tr}[\mathrm{d}\Pi(\lambda)\rho]$.

The Wigner function has the important property that it reproduces the expectation value of an observable $f(\zeta \hat{R})$ (see e.g. \cite[Sec.~4.2.3.]{degosson_harmonic_2017}):
Whenever the function $W_{\rho}W_{f(\zeta\hat{R})}$ is integrable, we have
	\begin{equation}\label{eq:wignerstatistical}
	\tr\left[\rho f\left(\zeta\hat{R}\right)\right]=\int W_{\rho}(v)W_{f(\zeta\hat{R})}(v)\mathrm{d} v,
	\end{equation}
	with $i\in \{1,...,2m\}$ and notating the admissible observables $\zeta\hat{R}:=\sum_i\zeta_i\hat{R}_i$.

We introduce the following notation: We use $\bullet$ for arguments of functions. 
E.g. $\sqrt{\bullet}$ denotes the function $x\mapsto \sqrt{x}$.

The Wigner function of a state is a special case of the Wigner transform that is well-defined as the inverse of the \textit{Weyl-quantization}, which for polynomially bounded functions $f:\mathbb{R}^{k}\to\mathbb{C}$ yields well defined operators $Q[f(\zeta_1\cdot\bullet,...,\zeta_k\cdot\bullet)]$ such that~\cite[Thm.~2.18]{folland_harmonic_1989}:
\begin{equation}\label{eq:weyl-inverse}
W_{Q[f]}=f.
\end{equation}
	
	\section{Hidden variable models and contextuality}
	Here we define non-contextual hidden variable models.
	These are value assignments $\lambda_{\varphi}$ for observables such as the homodyne measurements defined by $\zeta\hat{R}$.
	Contextuality is an assumption about their behaviour for joint measurements of commuting homodyne observables.
	We refer to a set
	\begin{equation}
	C=\left\{\zeta_1\hat{R},...,\zeta_k\hat{R}\right\}
	\end{equation}
	of pair-wise commuting homodyne observables as a \textit{context}.
	
	The following key lemma is analogous to~\cite[Lem.~1]{delfosse_wigner_2014}.
		\begin{lemma}\label{lemma:additive}
Consider a multiplicative function $\lambda:\mathbb{R}^{2m}\to\mathbb{R}$ such that $\lambda(\zeta_1+\zeta_2)=\lambda(\zeta_1)+\lambda(\zeta_2)$ for all $\zeta_1,\zeta_2\in \mathbb{R}^{2m}$ with $[\zeta_1,\zeta_2]=0$.
Then, $\lambda$ is a linear functional.
	\end{lemma}
	\begin{proof}
		Consider the canonical basis of $\mathbb{R}^{2m}$ that we denote by $(e_1,...,e_m,f_1,...,f_m)$.
		We define the planes $P_i:=\mathrm{span}(e_i,f_i)$. By the first assumption on value assignments, we immediately obtain
		\begin{equation}
		\lambda\left[\sum_{i=0}^n\alpha_ie_i+\beta_if_j\right]=\sum_{i=1}^n\lambda[\alpha_ie_i+\beta_i f_i].
		\end{equation}
		It thus suffices to work with the restrictions of $\lambda$ to the planes $P_i$. 
		We work with a second plane $P_j$ with $i\neq j$ and introduce the notation $u:=\alpha_i e_i, v:=\beta_i f_i$ and $u':=\beta_ie_j, v':=\alpha_if_j$. 
		We can decompose
		\begin{equation}
		u+v=\frac{1}{2}((u+v+u'+v')+(u+v-u'-v')),
		\end{equation}
		where both braced summands yield commuting operators:
		\begin{align*}
		[(u+&v+u'+v'),(u+v-u'-v')]\\
		& = [u,v]
		+[v,u]-[u',v']-[v',u']=0
		\end{align*}
		Similarly, we have that $[u\pm v',v\pm u']=0$.
		Combined, we obtain:
		\begin{align}
		\begin{split}
		\lambda(u+&v)=\\
		=&\lambda\left(\frac{1}{2}[(u+v+u'+v')+(u+v-u'-v')]\right)\\
		=&\frac{1}{2}\lambda(u+v+u'+v')+\frac{1}{2}\lambda(u+v-u'-v')\\
		=&\frac{1}{2}[\lambda(u+v')+\lambda(v+u')+\lambda(u-v')+\lambda(v-u')]\\
		=&\frac{1}{2}(\lambda(u)+\lambda(v')+\lambda(v)+\lambda(u')\\
		&+\lambda(u)+\lambda(v')-\lambda(v)-\lambda(u'))\\
		=&\lambda(u)+\lambda(v),
		\end{split}
		\end{align}
		which completes the proof.
	\end{proof}
Lemma~\ref{lemma:additive} has the following immediate consequence: Any assignments of values on operators that are additive on contexts can be identified with elements in $(\mathbb{R}^{2m})^*\cong \mathbb{R}^{2m}$.
Non-contextuality ensures this for hidden variable assignments and we can restrict to hidden variables being linear functionals.

Similarly to \cite{delfosse_wigner_2014}, we consider \textit{inferable observables} whose value can be directly inferred from observable quantities and classical postprocessing: namely, we call an observables $A$ ``inferable'' if $A=f(\zeta_1\hat{R},...,\zeta_k\hat{R})$ where
$\{\zeta_1\hat{R},...,\zeta_k\hat{R}\}$ is a context and $f:\mathbb{R}^k\to\mathbb{C}$ is a function.
	\begin{definition}[non-contextual HVM]\label{definition:HVM}
		Let $\rho$ be a density operator over $\mathcal{H}$ with well defined Wigner function as in Definition~\ref{definition:wignerfunction}.
		A hidden variable model for $\rho$ consists of a non-empty measurable space of hidden states $S$, a probability measure $\nu_{\rho}$ on $S$ and conditional probabilities $p_C(s|\varphi)$ for each context $C$, $\varphi \in S$ and $s\in \mathbb{R}^{k}$. 
		These are subject to the following conditions:
		\begin{enumerate}
			\item Consider a context
$ C=\{\zeta_1\hat{R},...,\zeta_k\hat{R}\}$.
			 For any internal state $\varphi$ there is an assignment $\lambda_{\varphi}$ such that all inferable observables $A$ have definite values $\lambda_{\varphi}(A)$ and
			\begin{equation}
			p_C(s|\lambda_{\varphi})=\prod_{i,\zeta_i \hat{R}\in C}\delta(s_i-\lambda_{\varphi}[\zeta_i \hat{R}]).
			\end{equation}
			\item Let $f:\mathbb{R}^{k}\to \mathbb{C}$ be a polynomial map.
			Then, it holds that
			\begin{equation}\label{eq:nc_condition}
			\lambda_{\varphi}\left[f\left(\zeta_1\hat{R},...,\zeta_k\hat{R}\right)\right]=f\left(\lambda_{\varphi}\left[\zeta_1\hat{R}\right],...,\lambda_{\varphi}\left[\zeta_k\hat{R}\right]\right).
			\end{equation}
			\item Recovering the predictions of quantum mechanics: For a smooth and bounded function, we require
			\begin{equation}\label{eq:recoverQM1}
			\tr\left[f\left(\zeta\hat{R}\right)\rho\right]=\int_S f\left(\lambda_{\varphi}\left[\zeta\hat{R}\right]\right)\mathrm{d}\nu_{\rho}(\varphi).
			\end{equation}
		\end{enumerate}
	\end{definition}
	
	A more operational definition of contextual hidden variable models can be formulated using the spectral theorem. 
	We give such a definition in Section~\ref{appendix:definition}, Definition~\ref{definition:HVM2} and prove equivalence to Definition~\ref{definition:HVM}.
	In the main text we work solely with Definition~\ref{definition:HVM} as it is most useful to show the equivalence with Wigner function negativity.
	
	Contextuality is defined by what it is not:
	\begin{definition}
		A state $\rho$ is \emph{contextual} if it does not admit a non-contextual hidden variable model.
	\end{definition}
	\section{Equivalence between Contextuality and Wigner non-negativity}\label{sec:Equivalence}
	In this section we prove that both previously defined notions of non-classicality are in fact equivalent. We first focus on prepare and measure quantum scenarios. We discuss the case for general quantum scenarios in Section~\ref{sect:FullGaussianOpticsAndSpekkens}.
	\begin{theorem}\label{theorem:main}
		Let $m\geq 2$ and let $\rho$ be a state. $\rho$ admits a non-contextual value assignments if and only if $W_{\rho}\geq 0$.
	\end{theorem}
	\begin{proof}[Proof of Theorem~\ref{theorem:main}]
		\underline{$W_{\rho}\geq 0$ $ \implies$ Def.~\ref{definition:HVM}}:
	 The non-negativity of the Wigner function implies the existence of a non-contextual value assignment as follows:
		The measure space of hidden variables $S$ can be taken to be the phase space $S=\mathbb{R}^{2m}$.
		Thus, we can set the measure $\nu_{\rho}=q_{\rho}\mu$ with  $q_{\rho}:=W_{\rho}$, where $\mu$ denotes the Lebesgue measure.
		$W_{\rho}$ is a probability distribution whenever $W_{\rho}$ is a positive function. 
		
		For all inferable obvservables $f(\zeta_1\hat{R},...,\zeta_k\hat{R})$ we set
		\begin{equation} \lambda_{\varphi}[f(\zeta_1\hat{R},...,\zeta_k\hat{R})]:= W_{f(\zeta_1\hat{R},...,\zeta_k\hat{R})}(\varphi).
		\end{equation}
		 This immediately defines the conditional probabilities $p_C(s|\varphi)$ as in Definition~\ref{definition:HVM}(1.). 
		 	By Eq.~\eqref{eq:wignerstatistical}, it holds that for an inferable observable $f(\zeta\hat{R})$ and in particular for $f$ bounded, we have
		\begin{equation}\label{eq:traceWigner}
		\tr\left[\rho f(\zeta\hat{R})\right]=\int W_{\rho}(v)W_{f(\zeta\hat{R})}(v)\mathrm{d}v.
		\end{equation}
		
        Finally non-contextuality follows from standard arguments for the Wigner function:
		Consider a context $C=\left\{\zeta_1\hat{R},...,\zeta_k\hat{R}\right\}$.
	   Then, the Weyl quantization yields by~\cite[Prop.~13.3]{hall_quantum_2013}
		\begin{align}
		\begin{split}
		Q\left(W_{\zeta_1\hat{R}}...W_{\zeta_{k}\hat{R}}\right)
		&=Q\left[(\zeta_1\cdot \bullet)...(\zeta_k\cdot \bullet)\right]\\
		&= \frac{1}{k!}\sum_{\sigma\in S_k}\zeta_{\sigma(1)}\hat{R}...\zeta_{\sigma(k)}\hat{R}\\
		&=\zeta_1\hat{R}...\zeta_k\hat{R},
		\end{split}
		\end{align}
		where we used that all operators $\zeta_i \hat{R}$ commute in the third equation.
		Taking the Wigner function of this expression combined with Eq.~\eqref{eq:weyl-inverse} implies 
		\begin{equation}
		W_{\zeta_1\hat{R}...\zeta_k\hat{R}}=W_{\zeta_1\hat{R}}...W_{\zeta_k\hat{R}}.
		\end{equation}
		Thus, we see that $W_{\bullet}$ respects multiplications for commuting homodyne operators.
		This can be directly extended to all polynomials, which immediately verifies Eq.~\eqref{eq:nc_condition}.
		\newline
		\underline{$W_{\rho}\geq 0 \impliedby$ Def.~\ref{definition:HVM}}: 
		Given a set of non-contextual value assignment for a state $\rho$, we can compute the corresponding Wigner functions:
		\begin{align}
		\begin{split}\label{eq:wignercalculation}
		W_{\rho}(\zeta)=&\frac{1}{2\pi}\int\chi_{\rho}(v)e^{\ii [v, \zeta]}\mathrm{d} v\\
		=& \frac{1}{2\pi}\int e^{\ii [v,\zeta]}\tr[D(v)\rho]\mathrm{d} v\\
		=& \frac{1}{2\pi}\int e^{\ii [v,\zeta]}\tr\left[e^{\ii(v^{T}\omega\hat{R})}\rho\right]\mathrm{d} v.\\
	=&\frac{1}{2\pi}\int e^{\ii [v, \zeta]}\left(\int_S e^{\ii\lambda_{\varphi}(v^{T}\omega )}\mathrm{d}\nu_{\rho}(\varphi)\right)\mathrm{d} v.
	\end{split}
		\end{align}
		We can use that every value assignment defines a linear functional by Lemma~\ref{lemma:additive} and Definition~\ref{definition:HVM}(2.). 
		Hence, we can write $\lambda_{\varphi}:\mathbb{R}^{2m}\to \mathbb{R}$ as $\lambda_{\varphi}(\zeta)= \varphi\cdot\zeta$.
		
		Assume that $\nu_{\rho}$ is absolutely continuous and can be written as $q_{\rho}\mu$, with $q_{\rho}\in L^1$ and $\mu$ the Lebesgue measure.
		Then, the Fourier inversion theorem implies that 
		\begin{align}
		\begin{split}
		W_{\rho}(\zeta)
		=&\frac{1}{2\pi}\int_{\mathbb{R}^{2m}}e^{\ii [v,\zeta]}\int_{\mathbb{R}^{2m}} e^{\ii [v,\varphi]} q_{\rho}(\varphi)\mathrm{d}\varphi\mathrm{d}v\\
		=&q_{\rho}(-\zeta)\geq 0.
		\end{split}
		\end{align}
		However, $\Phi_{\rho}:= \int_S e^{\mathrm{i}[v,\varphi]} \mathrm{d}\nu_{\rho}$ is (up to a minus sign in the exponent) the characteristic function of the probability measure $\nu_{\rho}$. 
	    We only work with states $\rho$ such that $W_{\rho}$ is well defined as a function and the integral 
		\begin{equation}
		   W_{\rho}=\int_{\mathbb{R}^{2m}}e^{\mathrm{i}[v,\zeta]} \Phi_{\rho}(v)\mathrm{d}v
		\end{equation}
		is only well defined if $\Phi_{\rho}$ is integrable. 
		By the inversion theorem for characteristic functions of probability measures, this implies that $\nu_{\rho}$ is absolutely continuous, which completes the proof.
	\end{proof}

	\section{An operational definition of contextuality}\label{appendix:definition}
	The following definition of non-contextual hidden variable models is a more operational one than Definition~\ref{definition:HVM} and only requires a HVM for the spectral projectors.
    Different contexts can implement the same family of projector valued measures and hence of the same measurements.
    Moreover, given a context $C$ one can define a measurement by considering a function $f$ of the outcomes that we used in previous sections to define inferable observables. Here, we call the map $f$ a \textit{post-processing map}. 
    The resulting projective measure of a set of outcomes $X$ can then be obtained as the integral over all values $x$ in the spectrum of the context $C$ such that $f(x)\in X$. 
    This projector is given by $\int_{f^{-1}(X)} \mathrm{d}\Pi_C$.
    A non-contextual hidden variable model is then required to assign the same probability measure to these spectral measures independently of the way the measurement was implemented via a post-processing map.
    This motivates the definition of a non-contextual hidden variable model for spectral measures below.
    Compare Ref.~\cite{delfosse_wigner_2014} for a discussion of an analogous definition for discrete systems.
	
	Denote by $\chi_X$ the characteristic function (or indicator function) of a set $X$:
	\begin{equation}
  \chi_{X}(x)=\begin{cases}
    1, & \text{if $x\in X$}.\\
    0, & \text{otherwise}.
  \end{cases}
\end{equation}

Given a state $\rho$ with well defined Wigner function as in Definition~\ref{definition:wignerfunction}.
A hidden variable model for $\rho$ consists of a non-empty measurable space of hidden states $S$ and a value assignment $\lambda_{\varphi}$ for all $\varphi\in S$. 
		The latter assigns to spectral measures $\Pi_{A}$ associated to inferable observables $A$ a probability measure $\lambda_{\varphi}[\Pi_{A}]$. 
		We assume \emph{outcome determinism}:
		\begin{equation}
		\lambda_{\varphi}\left[\Pi_{A}\right]=\delta(\bullet-\lambda_{\varphi}[A])
		\end{equation}
		or equivalently
		\begin{equation}
		    \lambda_{\varphi}\left[\Pi_{A}\right](X)=\chi_X(\lambda_{\varphi}[A])
		\end{equation}
		for some value $\lambda_{\varphi}[A]$.
		
	\begin{definition}[non-contextual HVM, projector version]\label{definition:HVM2}
		We call the above data a \textit{non-contextual hidden variable model} for $\rho$ if the following conditions hold:
		\begin{enumerate}
			\item Non-contextuality: Given a context $C=\{\zeta_1\hat{R},...,\zeta_k\hat{R}\}$ and a polynomially bounded post-processing map $f:\mathbb{R}^k\to\mathbb{C}$, then for all compact $X\subseteq \Spec[f(\zeta_1\hat{R},...,\zeta_k\hat{R})]$:
			\begin{equation}
		\lambda_{\varphi}\left[\int_{f^{-1}(X)}\mathrm{d}\Pi_C\right]
			=\int_{f^{-1}(X)}\mathrm{d}\lambda_{\varphi}(\Pi_C)
\label{eq:non-contextual-projector}.
			\end{equation}
			\item Recovering the predictions of quantum mechanics for compact $X$:
			\begin{equation}\label{eq:recoverQM2}
			\tr[\Pi_{\zeta\hat{R}}(X)\rho]=\int_S\lambda_{\varphi}[\Pi_{\zeta\hat{R}}](X)\mathrm{d}\nu_{\rho}(\varphi).
			\end{equation}
		\end{enumerate}
	\end{definition}

We first prove that the existence of a non-contextual HVM in the sense of Definition~\ref{definition:HVM2} implies the existence of a non-contextual HVM in the sense of Definition~\ref{definition:HVM}.
Then we prove that a positive Wigner function implies the existence of a non-contextual HVM in the sense of Definition~\ref{definition:HVM2}, which completes the argument.

\underline{Def.~\ref{definition:HVM2} $\implies$ Def.~\ref{definition:HVM}}. 
We first show that Eq.~\eqref{eq:non-contextual-projector} implies Eq.~\eqref{eq:nc_condition}.
Notice first that for each context $C$, and every integrable post-processing map $f$ the following holds:
\begin{align}
\begin{split}
    	\int_{f^{-1}(X)}\mathrm{d}\Pi_C&=\int \chi_{f^{-1}(X)}[\zeta_1,\ldots,\zeta_k]\mathrm{d}\Pi_C\\
    	&= \int \chi_X[f(\zeta_1,\ldots,\zeta_k]\mathrm{d}\Pi_C\\
    	&=\chi_X\left[f(\zeta_1\hat{R},\ldots,\zeta_k\hat{R}\right]\\
    	&=\Pi_{f(\zeta_1\hat{R},...,\zeta_k\hat{R})}(X).
    	\end{split}
\end{align}
Therefore, Eq.~\eqref{eq:non-contextual-projector} can equivalently be formulated as
\begin{equation}
	\lambda_{\varphi}\left[\Pi_{f(\zeta_1\hat{R},...,\zeta_k\hat{R})}(X)\right]=\chi_{f^{-1}(X)}(\lambda_{\varphi}[\zeta_1\hat{R}],...,\lambda_{\varphi}[\zeta_k\hat{R}]).
\end{equation}
	By definition of outcome determinism we have 
\begin{equation}
\lambda_{\varphi}\left[\Pi_{f(\zeta_1\hat{R},...,\zeta_k\hat{R})}(X)\right]=\chi_X(\lambda_{\varphi}[f(\zeta_1\hat{R},...,\zeta_k\hat{R})])
\end{equation}
for all $X\subset \Spec[f(\zeta_1\hat{R},...,\zeta_k\hat{R})]$.
 Combined with~\eqref{eq:non-contextual-projector} this implies
	\begin{align}
	\begin{split}
	\chi_X(\lambda_{\varphi}[f(\zeta_1\hat{R}&,...,\zeta_k\hat{R})])\\
	&=\chi_{f^{-1}(X)}(\lambda_{\varphi}[\zeta_1\hat{R}],...,\lambda_{\varphi}[\zeta_k\hat{R}])\\
	&=\chi_X[f(\lambda_{\varphi}[\zeta_1\hat{R}],...,\lambda_{\varphi}[\zeta_k\hat{R}])]
	\end{split}
	\end{align}
	for all $X\subset \Spec[f(\zeta_1\hat{R},...,\zeta_k\hat{R})]$.
	 By choosing $X=\{x\}$, this includes Eq.~\eqref{eq:nc_condition}.
	 In particular, this suffices to invoke Lemma~\ref{lemma:additive}: All $\lambda_{\varphi}$ are linear functions in inferable operators. 
	
	We obtain~\eqref{eq:recoverQM1} from~\eqref{eq:recoverQM2} as follows:
	\begin{align}
	\begin{split}
	\tr[f(\zeta\hat{R})\rho]&=\tr\left[\int f(x)\mathrm{d}\Pi_{\zeta\hat{R}}(x)\rho\right]\\
	&=\int f(x)\tr[\mathrm{d}\Pi_{\zeta\hat{R}}(x)\rho]\\
	&\stackrel{\eqref{eq:recoverQM2}}{=}\int f(x)\int_S\mathrm{d}\lambda_{\varphi}[\Pi_{\zeta\hat{R}}](x)\mathrm{d}\nu_{\rho}(\varphi)\\
	&=\int_S\lambda_{\varphi}\left[\int f(x)\mathrm{d}\Pi_{\zeta\hat{R}}(x)\right]\mathrm{d}\nu_{\rho}(\varphi)\\
	&=\int_S f(\lambda_{\varphi}[\zeta\hat{R}])\mathrm{d}\nu_{\rho}(\varphi).
	\end{split}
	\end{align}
	
\underline{$W_{\rho}\geq 0 \implies$ Def.~\ref{definition:HVM2}}.
	Showing that positivity of $W_{\rho}$ implies the existence of a non-contextual HVM in the sense of Definition~\ref{definition:HVM2} follows from well known facts about the Wigner function (see e.g. Ref.~\cite{folland_harmonic_1989}).
	In particular, we assume that the following statement is standard, but in abscence of a reference we added a proof in Appendix~\ref{appendix:prooflemma}.
\begin{lemma}\label{lemma:wignercomposition}
	Given a context $C=\{\zeta_1\hat{R},...,\zeta_k\hat{R}\}$ and a bounded function $f\in L^2(\mathbb{R}^k)$, then
	\begin{equation}\label{eq:multiplicationWigner}
	W_{f(\zeta_1\hat{R},...,\zeta_k\hat{R})}=f(W_{\zeta_1\hat{R}},...,W_{\zeta_k\hat{R}}).
	\end{equation}
\end{lemma}

	Thus, if $W_{\rho}\geq 0$, we can simply choose $S=\mathbb{R}^{2m}$, $q_{\rho}=W_{\rho}$ and $\lambda_{\varphi}(A)=W_{A}(\varphi)$. 
	As $f$ is a polynomial and $\mathbb{R}^k$ is Hausdorff, $f^{-1}(X)$ is compact and therefore $\chi_{f^{-1}(X)}\in L^2(\mathbb{R}^{k})$.
	Finally, by Lemma~\ref{lemma:wignercomposition}, this model satisfies
	\begin{align}
	\begin{split}
	\lambda_{\varphi}[\chi_{f^{-1}(X)}&(\zeta_1\hat{R},...,\zeta_k\hat{R})]\\
	&=W_{\chi_{f^{-1}(X)}(\zeta_1\hat{R},...,\zeta_k\hat{R})}(\varphi)\\
	&=\chi_{f^{-1}(X)}[W_{\zeta_1\hat{R}}(\varphi),...,W_{\zeta_k\hat{R}}(\varphi)]\\
	&=\chi_{f^{-1}(X)}(\lambda_{\varphi}[\zeta_1\hat{R}],...\lambda_{\varphi}[\zeta_k\hat{R}]).
	\end{split}
	\end{align}

	This is precisely the non-contextuality condition in Definition~\ref{definition:HVM2}.

\section{Extension to the full Gaussian quantum optics subtheory}
\label{sect:FullGaussianOpticsAndSpekkens}

Above, we have restricted ourselves to prepare-and-measure quantum scenarios comprising homodyne measurements on continuous-variable quantum states. This scenarios may seem restrictive, since many experiments of interest quantum optics and continuous-variable quantum information  include other types of operations, such as Gaussian POVMs or Gaussian CPTP maps \cite{weedbrook_gaussian_2012}. However, as we now argue, our equivalence result extends to more general operations. In particular, it also holds that Wigner negativity is equivalent to contextuality with respect to the full subtheory of bosonic Gaussian optics, wherein the allowed quantum state operations are preparations of Gaussian states,   measurements of Gaussian POVMs and implementations of Gaussian CPTP maps \cite{giedke_characterization_2002,mari_wigner_2012,weedbrook_gaussian_2012,bartlett_reconstruction_2012}. In this larger setting, POVM measurements are not always sharp (i.e., they do not correspond to standard projective measurement in the quantum formalism) and we also consider transformations. One can study non-contextual models for this setting using Spekkens' generalized framework~\cite{spekkens2005contextuality} of contextuality. It is known that the Wigner function provides a non-contextual hidden variable model for this full Gaussian quantum subtheory \cite{bartlett_reconstruction_2012} . Since Gaussian operations include our previous prepare-and-measure scenarios as a subset, our result immediately carries over: noncontextuality with respect to full Gaussian operations, implies Wigner negativity. Furthermore, our result implies the Wigner function is the unique noncontextual hidden variable model for the Bosonic Gaussian subtheory.

Sequential contextuality, a notion of contextuality for transformations different to Spekkens', has been recently introduced in \cite{Mansfield18_Quantum_Advantage}. Therein, an ontological model is noncontextual if compositions of transformations at the quantum level are mapped to compositions of transformations at the ontological level. In this sense, Bosonic Gaussian quantum optics remains non-contextual. This is easily seen from the fact that the Wigner function $W_{\mathcal{E}}(z,z)$ of a CPTP map $\mathcal{E}$ satisfies a well-behaved composition rule $W_{\mathcal{E}_2\circ \mathcal{E}_1}=W_{\mathcal{E}_2}W_{\mathcal{E}_1}$: this is easily seen using \cite[Eq. (95-99)]{bartlett_reconstruction_2012}.

\section{Conclusion}

In this work we established an equivalence between the negativity of the Wigner function and contextuality of hidden variable models containing homodyne measurements. 
Moreover, we showed how this equivalence can be lifted to contextuality in the general theory of Gaussian quantum optics. As a consequence, there is a clean link between efficient classical simulability, noncontextuality and the existence of a distinguished nonnegative quasiprobability representation: the traditional Wigner function~\cite{wigner1997quantum}.
 We generalized a proof technique developed for discrete, odd dimensional quantum systems~\cite{delfosse_contextuality_2015} to the continuous variable setting, which contains subtleties regarding the integrability of functions.

\emph{Open questions.} Our work does not apply to discrete-variable systems of even dimension. In particular, the presence of state-idependent contextuality with respect to Pauli observables \cite{Mermin_RevModPhys.65.803}  implies that the multiqubit stabilizer subtheory is contextual. Thus the result of \cite{howard_contextuality_2014} cannot be generalized. It is known, though, that for any family of stabilizer operations that is free of state independent contextuality, a classical simulation method based on a matching non-contextual positive Wigner representation exists, but this does not apply to the full stabilizer subtheory \cite{bermejo-vega_contextuality_2017}. It is known that there exist generalized Wigner functions that describe the entire stabilizer subtheory and enable efficient classical simulability \cite{raussendorf_phase-space-simulation_2020}. However, they are over-complete ones, and the associated hidden variable models exhibit certain types of contextuality, e.g., Spekkens generalized contextuality for preparations and transformations~\cite{spekkens2005contextuality}. The precise role between classical simulability, negative quasi-probability, contextuality and possibly other notions of classicality (such as, e.g., psi-epistemicity~\cite{lillystone2019contextual}) remains to be understood. 

Another natural open question is to prove a robust version of the above equivalence. 
Is a hidden variable model that is approximately non-contextual in a suitable sense equivalent to a Wigner function with little negativity, e.g. quantified by the $L^1$-norm of the Wigner function? We point out that, for finite odd-dimensional quantum systems, a robust Hudson's theorem that signals out  approximate pure stabilizer states as those that are approximately Wigner-nonnegative exists~\cite{gross2017schur}.

    \section{Acknowledgments}
    JH was funded by the Foundational Questions Institute (FQXi). JBV acknowledges funding from the European Union’s Horizon 2020 research and innovation programme under the Marie Skłodowska-Curie grant agreement Nº 754446 and UGR Research and Knowledge Transfer Found – Athenea3i. 
    
    Upon completion of this work, we became aware of Ref~\cite{booth2021contextuality}, which makes an essentially equivalent claim also based on measure theoretic techniques. We are grateful to Robert I. Booth, Ulysse Chabaud and Pierre-Emmanuel Emeriau for comments and fruitful cooperation. We thank Jens Eisert for detailed comments on the manuscript.
    
    \bibliography{equivalence_wigner}
	
	\onecolumngrid
	\appendix

\section{Proof of Lemma~\ref{lemma:wignercomposition}}\label{appendix:prooflemma}
We need a slight modification of an argument from~\cite[Eq.~(2.18), (2.19)]{folland_harmonic_1989} based on the following lemma, which is a special case of~\cite[Thm.~2.15]{folland_harmonic_1989}.
Recall that we denote the Weyl quantization scheme by $Q$. 
\begin{lemma}
	Let $A\in \mathrm{Sp}(n,\mathbb{R})$. 
	There is a unitary $U_{A}$ such that
	\begin{equation}
	Q[f\circ A]=U_{A} Q[f]U_{A}^{\dagger},
	\end{equation}
	for all $f\in L^2(\mathbb{R}^{2m})$.
\end{lemma}
\begin{proof}[Proof of Lemma~\ref{lemma:wignercomposition}]
	Consider $\zeta_1,...,\zeta_k$ such that $[\zeta_i,\zeta_j]=0$, then it was shown in Ref.~\cite[Lem.~17]{folland_harmonic_1989} that there is an $A\in\mathrm{Sp}(n,\mathbb{R})$ such that $\zeta_i A=e_i$.
	Applying Lemma~\ref{lemma:wignercomposition} yields
	\begin{align*}
	&Q[f(\zeta\cdot \bullet,...,\zeta_k\cdot \bullet)]\\
	&=Q[(e_1\cdot A \bullet,...,e_k\cdot A \bullet)]\\
	&=Q[f(e_1\cdot \bullet,...,e_k \cdot \bullet)\circ A]\\
	&=U_{A}Q[f(e_1\cdot \bullet,...,e_k \cdot \bullet]U_{A}^{\dagger}.
	\end{align*}
	Eq.~\eqref{eq:multiplicationWigner} is implied for $f(e_1\cdot \bullet,...,e_k \cdot \bullet)$ since the spectral functional calculus and Weyl-quantization yield the same multiplication operator.
	What is more, we can check that
	\begin{equation}
	U_{A}f(e_1\hat{R},...,e_k\hat{R})	U_{A}^{\dagger}=f(\zeta_1\hat{R},...,\zeta_k\hat{R}).
	\end{equation}
	Thus
	\begin{equation}
	f(\zeta_1\hat{R},...,\zeta_k\hat{R})=Q[f(\zeta_1\cdot \bullet,...,\zeta_k \cdot \bullet)]
	\end{equation}
	which implies
	\begin{equation}
	W_{f(\zeta_1\hat{R},...,\zeta_k\hat{R})}\\= W_{Q[f(\zeta_1\cdot \bullet,...,\zeta_k \cdot \bullet)]}=f(W_{\zeta\hat{R}},...,W_{\zeta_k\hat{R}}).
	\end{equation}
\end{proof}

\end{document}